\documentclass[manuscript,nonacm]{acmart}

\usepackage{amsfonts,amssymb}
\usepackage{graphicx}
\usepackage{color}
\usepackage{todonotes}
\usepackage{algorithm}
\usepackage{algorithmic}
\usepackage{xspace}
\usepackage{etoolbox}
\usepackage{xcolor}
\usepackage{multirow}
\usepackage{enumitem}
\usepackage[normalem]{ulem}
\usepackage{bm}
\usepackage{subcaption}

\usepackage{balance}

\usepackage{mathtools}

\newtheorem{definition}{Definition}
\newtheorem{lemma}{Lemma}
\newtheorem{theorem}{Theorem}

\usepackage{tikz}
\usetikzlibrary{matrix, decorations, patterns, positioning, shapes, 3d, calc, intersections, arrows, fit}

\usepackage{cleveref}
\Crefname{definition}{Definition}{Definitions}
\crefname{definition}{Def.}{Defs.}
\Crefname{theorem}{Theorem}{Theorems}
\crefname{theorem}{Theorem}{Theorems}
\Crefname{lemma}{Lemma}{Lemmas}
\crefname{lemma}{Lemma}{Lemmas}
\Crefname{section}{Section}{Sections}
\crefname{section}{\S}{\S}
\crefname{algorithm}{Alg.}{Algs.}
\Crefname{algorithm}{Algorithm}{Algorithms}
\crefname{figure}{Fig.}{Figs.}
\Crefname{figure}{Figure}{Figures}
\crefname{table}{Tab.}{Tabs.}
\Crefname{table}{Table}{Tables}

\newcommand{\tmp}[1]{q_{prev}}

\newcommand{\Real}{\mathbb{R}}

\definecolor{pastelviolet}{rgb}{0.8, 0.6, 0.79}
\definecolor{babyblueeyes}{rgb}{0.63, 0.79, 0.95}
\definecolor{pastelyellow}{rgb}{0.99, 0.99, 0.59}
\definecolor{pastelgreen}{rgb}{0.47, 0.87, 0.47}
\definecolor{pastelred}{rgb}{1.0, 0.41, 0.38}
\colorlet{patternblue}{blue!60}

\newcommand{\A}{\mathbf{A}}
\newcommand{\B}{\mathbf{B}}
\newcommand{\CC}{\mathbf{C}}
\newcommand{\D}{\mathbf{D}}
\newcommand{\x}{x_1}
\newcommand{\y}{x_2}
\newcommand{\z}{x_3}

\newcommand{\vc}[1]{\bm{#1}}

\definecolor{pdfurlcolor}{rgb}{0,0,0.6}
\definecolor{pdfcitecolor}{rgb}{0,0.6,0}
\definecolor{pdflinkcolor}{rgb}{0.6,0,0}

\begin{document}

\title{Tight Memory-Independent Parallel Matrix Multiplication Communication Lower Bounds}


\author{Hussam Al Daas}
\orcid{0000-0001-9355-4042}
\affiliation{%
  \institution{Rutherford Appleton Laboratory}
  \city{Didcot}
  \state{Oxfordshire}
  \country{UK}
}
\email{hussam.al-daas@stfc.ac.uk}

\author{Grey Ballard}
\orcid{0000-0003-1557-8027}
\affiliation{%
  \institution{Wake Forest University}
  \city{Winston-Salem}
  \state{NC}
  \country{USA}}
\email{ballard@wfu.edu}

\author{Laura Grigori}
\orcid{0000-0002-5880-1076}
\affiliation{%
  \institution{Inria Paris}
  \city{Paris}
  \country{France}}
\email{laura.grigori@inria.fr}

\author{Suraj Kumar}
\affiliation{%
  \institution{Inria Paris}
  \city{Paris}
  \country{France}
}
\email{suraj.kumar@inria.fr}

\author{Kathryn Rouse}
\affiliation{%
 \institution{Inmar Intelligence}
 \city{Winston-Salem}
 \state{NC}
 \country{USA}
}
\email{kathryn.rouse@inmar.com}


\begin{abstract}
Communication lower bounds have long been established for matrix multiplication algorithms.
However, most methods of asymptotic analysis have either ignored the constant factors or not obtained the tightest possible values.
Recent work has demonstrated that more careful analysis improves the best known constants for some classical matrix multiplication lower bounds and helps to identify more efficient algorithms that match the leading-order terms in the lower bounds exactly and improve practical performance.
The main result of this work is the establishment of memory-independent communication lower bounds with tight constants for parallel matrix multiplication.
Our constants improve on previous work in each of three cases that depend on the relative sizes of the aspect ratios of the matrices.
\end{abstract}



\maketitle

\section{Introduction}
\label{sec:intro}

The cost of communication relative to computation continues to grow, so the time complexity of an algorithm must account for both the computation it performs and the data that it communicates.
Communication lower bounds for computations set targets for efficient algorithms and spur algorithmic development.
Matrix multiplication is one of the most fundamental computations, and its I/O complexity on sequential machines and parallel communication costs have been well studied over decades \cite{HK81,ACS90,ITT04,SS14,DE+13}.

The earliest results established asymptotic lower bounds, ignoring constant factors and lower order terms.
Because of the ubiquity of matrix multiplication in high performance computations, more recent attempts have tightened the analysis to obtain the best constant factors \cite{SLLvdG19,KK+19,OLPSR20}.
These improvements in the lower bound also helped identify the best performing sequential and parallel algorithms that can be further tuned for high performance in settings where even small constant factors make significant differences.
We review these advances and other related work in \cref{sec:related}.

The main result of this paper is the establishment of tight constants for memory-independent communication lower bounds for parallel classical matrix multiplication.
In the context of a distributed-memory parallel machine model (see \cref{sec:prelim:costModel}), these bounds apply even when the local memory is infinite, and they are the tightest bounds in many cases when the memory is limited.
Demmel et al.~\cite{DE+13} prove asymptotic bounds for general rectangular matrix multiplication and show that three different bounds are asymptotically tight in separate cases that depend on the relative sizes of matrix dimensions and the number of processors.
Our main result, \cref{thm:main} in \cref{sec:main}, reproduces those asymptotic bounds and improves the constants in all three cases.
Further, in \cref{sec:algo}, we analyze a known algorithm to show that it attains the lower bounds exactly, proving that the constants we obtain are tight.
We present a comparison to previous work in \cref{tab:summary} and discuss it in detail in \cref{sec:conclusion}.

We believe one of the main features of our lower bound result is the simplicity of the proof technique, which makes a unified argument that applies to all three cases.
The key idea is to cast the lower bound as the solution to a constrained optimization problem (see \cref{lem:opt}) whose objective function is the sum of variables that correspond to the amount of data of each matrix required by a single processor's computation.
The constraints include the well-known Loomis-Whitney inequality \cite{LW49} as well as new lower bounds on individual array access (see \cref{lem:projlb}) that are necessary to establish separate results for the three cases.
All of the complexity of the three cases, including establishing the thresholds between cases and the leading terms in each case, are confined to a single optimization problem.
We use fundamental results from convex optimization (reviewed in \cref{sec:pastResults}) to solve the problem analytically.
This unified argument is elegant, it improves previous results to obtain tight constants, and it can be applied more generally to other computations that have iteration spaces with uneven dimensions.

\section{Related Work}
\label{sec:related}

\subsection{Memory-Dependent Bounds for Matrix Multiplication}
\label{sec:related:seq}

The first communication lower bound for matrix multiplication was established by Hong and Kung \cite{HK81}, who obtain the result using computation directed acyclic graph (CDAG) analysis that multiplication of square $n\times n$ matrices on a machine with cache size $M$ requires $\Omega(n^3/\sqrt M)$ words of communication.
Irony, Toledo, and Tiskin \cite{ITT04} reproduce the result using a geometric proof based on the Loomis-Whitney inequality (\cref{lem:LW}), show it applies to general rectangular matrices (replacing $n^3$ with $n_1n_2n_3$), and obtain an explicit constant of $(1/2)^{3/2}\approx .35$.
They also observe that the result is easily extended to the distributed memory parallel computing model (as described in \cref{sec:prelim:costModel}) under mild assumptions by dividing the bound by the number of processors $P$.
We refer to such bounds as ``memory-dependent,'' following \cite{BDHLS12-SS}, where the cache size $M$ is interpreted as the size of each processor's local memory.
Later, Dongarra et al. \cite{DPRSV08} tightened the constant for sequential and parallel memory-dependent bounds to $(3/2)^{3/2}\approx 1.84$.
More recently, Smith et al. \cite{SLLvdG19} prove the constant of 2 and show that it is obtained by an optimal sequential algorithm and is therefore tight.
Both of these results are proved using the Loomis-Whitney inequality.
Kwasniewski et al. \cite{KK+19} use CDAG analysis to obtain the same constant of 2 and show that it is tight in the parallel case (when memory is limited) by providing an optimal algorithm.

\subsection{Bounds for Other Computations}

Hong and Kung's proof technique can be applied to a more general set of computations, including the FFT \cite{HK81}.
Ballard et al. \cite{BDHS12} use the proof technique of \cite{ITT04} to generalize the lower bound (with the same explicit constant) to other linear algebra computations such as LU, Cholesky, and QR factorizations.
The constants of the lower bounds for these and other computations are tightened by Olivry et al. \cite{OLPSR20}, including reproducing the constant of 2 for matrix multiplication.
Kwasniewski et al. \cite{KK+21} also obtain tighter constants for LU and Cholesky factorizations using CDAG analysis.
Christ et al. \cite{CDKSY13} show that a generalization of the Loomis-Whitney inequality can be used to prove communication lower bounds for a much wider set of computations, but the asymptotic bounds do not include explicit constants.
This approach is applied to a tensor computation by Ballard and Rouse \cite{BKR18,BR20}.

\subsection{Memory-Independent Bounds for Parallel Matrix Multiplication}

This section describes the related work that focuses on the topic of this paper.
Aggarwal, Chandra, and Snir \cite{ACS90} extend Hong and Kung's result for matrix multiplication to the LPRAM parallel model, which closely resembles the model we consider with the exception that there exists a global shared memory where the inputs initially reside and where the output must be stored at the end of the computation.
Communication bounds in the related BSP parallel model are also memory independent, and Scquizzato and Silvestri \cite{SS14} establish the same asymptotic lower bounds for matrix multiplication in that model.
In addition to proving bounds for sequential matrix multiplication and the associated memory-dependent bound for parallel matrix multiplication, Irony, Toledo, and Tiskin \cite{ITT04} prove also that a parallel algorithm must communicate $\Omega(n^2/P^{2/3})$ words, and they provide explicit constants in their analysis.
Note that the size of the local memory $M$ does not appear in this bound.
Ballard et al. \cite{BDHLS12-SS} reproduce this result for classical matrix multiplication as well as prove similar results for fast matrix multiplication.
They distinguish between memory-dependent bounds (results described in \cref{sec:related:seq}) and memory-independent bounds for parallel algorithms, and they show the two bounds relate and affect strong scaling behavior.
In particular, when $M\gg n^2/P^{2/3}$ (or equivalently $P \gg n^3/M^{3/2}$), the memory-dependent bound is unattainable because the memory-independent bound is larger.
Demmel et al. \cite{DE+13} extend the memory-independent results to the rectangular case (multiplying matrices of dimensions $n_1\times n_2$ and $n_2\times n_3$), showing that three different bounds apply that depend on the relative sizes of the three dimensions and the number of processors, and their proof provides explicit constants.
For one of the cases, and for a restricted class of parallelizations, Kwasniewski et al. \cite{KK+19} prove a tighter constant and show that it is attainable by an optimal algorithm.

We summarize the constants obtained by these previous works and compare them to our results in \cref{tab:summary}.
Further details of the comparison are given in \cref{sec:comparison}.
Following \cref{thm:main}, the table assumes $m=\max\{n_1,n_2,n_3\}$, $n=\text{median}\{n_1,n_2,n_3\}$, and $k=\min\{n_1,n_2,n_3\}$.

\begin{table}
\begin{tabular}{|c|ccc|}
\hline
& $1\leq P\leq \frac mn$ & $\frac mn \leq P \leq \frac{mn}{k^2}$ & $\frac{mn}{k^2} \leq P$ \\
\hline
Leading term & $nk$ & $\left(\frac{mnk^2}{P}\right)^{1/2}$ & $\left(\frac{mnk}{P}\right)^{2/3}$ \\
\hline
\cite{ACS90} & - & - & $\left(\frac 12\right)^{2/3}\approx .63$ \\
\cite{ITT04} & - & - & $\frac 12=.5$ \\
\cite{DE+13} & $\frac{16}{25}=.64$  & $\left(\frac 23\right)^{1/2}\approx .82$ & 1 \\
\cref{thm:main} & 1 & 2 & 3 \\
\hline
\end{tabular}
\caption{Summary of explicit constants of leading term of parallel memory-independent rectangular matrix multiplication communication lower bounds for multiplication dimensions $m\geq n\geq k$ and $P$ processors}
\label{tab:summary}
\end{table}

\subsection{Communication-Optimal Parallel Matrix Multiplication Algorithms}
\label{sec:relatedalgs}

Both theoretical and practical algorithms that attain the communication lower bounds have been proposed for various computation models and implemented on many different types of parallel systems.
The idea of ``3D algorithms'' for matrix multiplication was developed soon after communication lower bounds were established; see \cite{B89,ACS90,Johnsson93,ABGJP95} for a few examples.
These algorithms partition the 3D iteration space of matrix multiplication in each of the three dimensions and assign subblocks across a 3D logical grid of processors.
McColl and Tiskin \cite{MT99} and Demmel et al.~\cite{DE+13} present recursive algorithms that effectively achieve similar 3D logical processor grid for square and general rectangular problems, respectively.
High-performance implementations of these algorithms on today's supercomputers demonstrate that these algorithms are indeed practical and outperform standard library implementations \cite{SD11,KK+19}.

\section{Preliminaries}
\label{sec:prelim}

\subsection{Parallel Computation Model}
\label{sec:prelim:costModel}

We consider the $\alpha$-$\beta$-$\gamma$ parallel machine model \cite{Thakur:CollectiveCommunications:2005,Chan:CollectiveCommunications:2007}.
In this model, each of $P$ processors has its own local memory of size $M$ and can compute only with data in its local memory.
The processors can communicate data to and from other processors via messages that are sent over a fully connected network (i.e., each pair of processors has a dedicated link so that there is no contention on the network).
Further, we assume the links are bidirectional so that a pair of processors can exchange data with no contention.
Each processor can send and receive at most one message at the same time.
The cost of communication is a function of two parameters $\alpha$ and $\beta$, where $\alpha$ is the per-message latency cost and $\beta$ is the per-word bandwidth cost.
A message of $w$ words sent from one processor to another costs $\alpha +\beta w$.
The parameter $\gamma$ is the cost to perform a single arithmetic operation.
For dense matrix multiplication when sufficiently large local memory is available, bandwidth cost nearly always dominates latency cost. Hence, we focus on the bandwidth cost in this work.
In our model, the communication cost of an algorithm is counted along the critical path of the algorithm so that if two pairs of processors are communicating messages simultaneously, the communication cost is that of the largest message.
In this work, we focus on memory-independent analysis, so the local memory size $M$ can be assumed to be infinite.
We consider limited-memory scenarios in \cref{sec:mm:mdb}.

\subsection{Fundamental Results}
\label{sec:pastResults}

In this section we collect the fundamental existing results we use to prove our main result, \cref{thm:main}.
The first lemma is a geometric inequality that has been used before in establishing communication lower bounds for matrix multiplication \cite{ITT04,BDHS12,DE+13}.
We use it to relate the computation performed by a processor to the data it must access.

\begin{lemma}[Loomis-Whitney \cite{LW49}]
\label{lem:LW}
Let $V$ be a finite set of lattice points in $\mathbb{R}^3$, i.e., points $(i,j,k)$ with integer coordinates. Let $\phi_i(V)$ be the projection of $V$ in the $i$-direction, i.e., all points $(j,k)$ such that there exists an $i$ so that $(i,j,k) \in V$. 
Define $\phi_j(V)$ and $\phi_k(V)$ similarly. 
Then 
$$|V| \leq |\phi_i(V)| \cdot |\phi_j(V)| \cdot |\phi_k(V)|,$$
where $|\cdot|$ denotes the cardinality of a set.
\end{lemma}

The next set of definitions and lemmas allow us to solve the key constrained optimization problem (\cref{lem:opt}) analytically.
We first remind the reader of the definitions of differentiable convex and quasiconvex functions and of the Karush-Kuhn-Tucker (KKT) conditions.
Here and throughout, we use boldface to indicate vectors and matrices and subscripts to index them, so that $x_i$ is the $i$th element of $\vc{x}$, for example.

\begin{definition}[{\cite[eq. (3.2)]{BV04}}]
\label{def:convex}
A differentiable function $f:\Real^d\rightarrow \Real$ is \emph{convex} if its domain is a convex set and for all $\vc{x},\vc{y} \in \textbf{dom} \; f$, 
$$f(\vc{y}) \geq f(\vc{x}) + \langle \nabla f(\vc{x}), \vc{y} - \vc{x} \rangle.$$
\end{definition}

\begin{definition}[{\cite[eq. (3.20)]{BV04}}]
\label{def:quasiconvex}
A differentiable function $g:\Real^d\rightarrow \Real$ is \emph{quasiconvex} if its domain is a convex set and for all $\vc{x},\vc{y} \in \textbf{dom} \; g$, 
$$g(\vc{y}) \leq g(\vc{x}) \text{ implies that } \langle \nabla g(\vc{x}), \vc{y} - \vc{x} \rangle \leq 0.$$
\end{definition}

\begin{definition}[{\cite[eq. (5.49)]{BV04}}]
\label{def:KKT}
Consider an optimization problem of the form 
\begin{equation}
\label{eq:optprob}
\min_{\vc{x}} f(\vc{x}) \quad \text{ subject to } \quad \vc{g}(\vc{x}) \leq \vc{0}
\end{equation} 
where $f:\Real^d \rightarrow \Real$ and $\vc{g}:\Real^d\rightarrow \Real^c$ are both differentiable. 
Define the dual variables $\vc{\mu}\in\mathbb{R}^c$, and let $\vc{J}_{\vc{g}}$ be the Jacobian of $\vc{g}$.
The \emph{Karush-Kuhn-Tucker (KKT)} conditions of $(\vc{x},\vc{\mu})$ are as follows:
\begin{itemize}
	\item \emph{Primal feasibility}: $\vc{g}(\vc{x}) \leq \vc{0}$;
	\item \emph{Dual feasibility}: $\vc{\mu} \geq 0$;
	\item \emph{Stationarity}: $\nabla f(\vc{x}) + \vc{\mu} \cdot \vc{J}_{\vc{g}}(\vc{x}) = \vc{0}$;
	\item \emph{Complementary slackness}: $\mu_i g_i(\vc{x})=0$ for all $i\in \{1,\dots,c\}$. 
\end{itemize}
\end{definition}

The next two results establish that our key optimization problem in \cref{lem:opt} can be solved analytically  using the KKT conditions.
While the results are not novel, we provide proofs for completeness.

\begin{lemma}[{\cite[Lemma 2.2]{BR20}}]
\label{lem:quasiconvex}
	The function $g_0(\vc{x}) = L - x_1x_2x_3$, for some constant $L$, is quasiconvex in the positive octant.
\end{lemma}
\begin{proof}
Let $\vc{x},\vc{y}$ be points in the positive octant with $g_0(\vc{y}) \leq g_0(\vc{x})$.
Then $y_1y_2y_3 \geq x_1x_2x_3$.
Applying the inequality of arithmetic and geometric means (AM-GM inequality) to the values $y_1/x_1$, $y_2/x_2$, $y_3/x_3$ (which are all positive), we have
\begin{equation}
\label{eq:AMGM}
\frac13 \left(\frac{y_1}{x_1}+\frac{y_2}{x_2}+\frac{y_3}{x_3}\right) \geq \left(\frac{y_1y_2y_3}{x_1x_2x_3}\right)^{1/3} \geq 1.
\end{equation}
Then $\nabla g_0(\vc{x}) = \begin{bmatrix} -x_2x_3 & -x_1x_3 & -x_1x_2 \end{bmatrix}$, and 
\begin{align*}
\langle \nabla g_0(\vc{x}), \vc{y}-\vc{x} \rangle &= 3x_1x_2x_3-y_1x_2x_3 - x_1y_2x_3 - x_1x_2y_3 \\
&= 3x_1x_2x_3 \left(1 - \frac13 \left(\frac{y_1}{x_1} + \frac{y_2}{x_2} + \frac{y_3}{x_3}\right)\right) \\
&\leq 0,
\end{align*}
where the last inequality follows from \cref{eq:AMGM}.
Then by \cref{def:quasiconvex}, $g_0$ is quasiconvex on the positive octant.
\end{proof}

\begin{lemma}
\label{lem:KKT}
Consider an optimization problem of the form given in \cref{eq:optprob}.
If $f$ is a convex function and each $g_i$ is a quasiconvex function, then the KKT conditions are sufficient for optimality.	
\end{lemma}
\begin{proof}
Suppose $\vc{x}^*$ and $\vc{\mu}^*$ satisfy the KKT conditions given in \cref{def:KKT}. 
If $\vc{\mu}^*=\vc{0}$, then by stationarity, $\nabla f(\vc{x}^*)=\vc{0}$.
Then the convexity of $f$ (\cref{def:convex}) implies 
$$f(\vc{x}) \geq f(\vc{x}^*) + \langle \nabla f(\vc{x}^*), \vc{x}-\vc{x}^* \rangle = f(\vc{x}^*)$$
for all $\vc{x}\in \textbf{dom} \; f$, which implies that $\vc{x}^*$ is a global optimum.

Now suppose $\vc{\mu}^*\neq\vc{0}$, then without loss of generality (and by dual feasibility) there exists $m\leq c$ such that $\mu^*_i>0$ for $1\leq i\leq m$ and $\mu^*_i=0$ for $m< i\leq c$.
Complementary slackness implies that $g_i(\vc{x}^*)=0$ for $1\leq i\leq m$.
Consider any (primal) feasible $\vc{x}\in \textbf{dom} \; f$.
Then $g_i(\vc{x})\leq 0$ for all $i$, and thus $g_i(\vc{x})\leq g_i(\vc{x}^*)$ for $1\leq i \leq m$.
By quasiconvexity of $g_i$ (\cref{def:quasiconvex}), this implies 
$$\langle \nabla g_i(\vc{x}^*), \vc{x}-\vc{x}^* \rangle \leq 0.$$
Stationarity implies that $\nabla f(\vc{x}^*) = -\sum_{i=1}^m \mu^*_i \nabla g_i(\vc{x}^*)$, and thus
$$\langle \nabla f(\vc{x}^*),\vc{x}-\vc{x}^* \rangle = -\sum_{i=1}^m \mu^*_i \langle \nabla g_i(\vc{x}^*), \vc{x}-\vc{x}^* \rangle \geq 0.$$
By convexity of $f$ (\cref{def:convex}), we therefore have
$$f(\vc{x}) \geq f(\vc{x}^*) + \langle \nabla f(\vc{x}^*), \vc{x}-\vc{x}^* \rangle \geq f(\vc{x}^*),$$
and thus $\vc{x}^*$ is a global optimum.
\end{proof}

\section{Main Lower Bound Result}
\label{sec:main}

\subsection{Lower Bounds on Individual Array Access}
\label{sec:projlb}

The following lemma establishes lower bounds on the number of elements of each individual matrix a processor must access based on the number of computations a given element is involved with.
This result is used to establish a set of constraints in the key optimization problem used in the proof of \cref{thm:main}.

\begin{lemma}
	\label{lem:projlb}
	Given a parallel matrix multiplication algorithm that multiplies an $n_1\times n_2$ matrix $\A$ by an $n_2\times n_3$ matrix $\B$ using $P$ processors, any processor that performs at least $1/P$th of the scalar multiplications must access at least $n_1n_2/P$ elements of $\A$ and at least $n_2n_3/P$ elements of $\B$ and also compute contributions to at least $n_2n_3/P$ elements of $\CC=\A\cdot \B$.
\end{lemma}
\begin{proof}
	The total number of scalar multiplications that must be computed is $n_1n_2n_3$.
	Consider a processor that computes at least $1/P$th of these computations.
	Each element of $\A$ is involved in $n_3$ multiplications.
	If the processor accesses fewer than $n_1n_2/P$ elements of $\A$, then it would perform fewer than $n_3 \cdot n_1n_2/P$ scalar multiplications, which is a contradiction.
	Likewise, each element of $\B$ is involved in $n_1$ multiplications.
	If the processor accesses fewer than $n_2n_3/P$ elements of $\B$, then it would perform fewer than $n_1 \cdot n_2n_3/P$ scalar multiplications, which is a contradiction.
	Finally, each element of $\CC$ is the sum of $n_2$ scalar multiplications.
	If the processor computes contributions to fewer than $n_1n_3/P$ elements of $\CC$, then it would perform fewer than $n_2\cdot n_1n_3/P$ scalar multiplications, which is again a contradiction.
\end{proof}

\subsection{Key Optimization Problem}

The following lemma is the crux of the proof of our main result (\cref{thm:main}).
We state the optimization problem abstractly here, but it may be useful to have the following intuition: the variable vector $\vc{x}$ represents the sizes of the projections of the computation assigned to a single processor onto the three matrices, where $\x$ corresponds to the smallest matrix and $\z$ corresponds to the largest matrix.
In order to design a communication-efficient algorithm, we wish to minimize the sum of the sizes of these projections subject to the constraints of matrix multiplication (and the processor performing $1/P$th of the computation), as specified by the Loomis-Whitney inequality (\cref{lem:LW}) and \cref{lem:projlb}.
A more rigorous argument that any parallel matrix multiplication algorithm is subject to these constraints is given in \cref{thm:main}.

We are able to solve this optimization problem analytically using properties of convex optimization (\cref{lem:KKT}).
The three cases of the solution correspond to how many of the individual variable constraints are tight.
When none of them is tight, we can minimize the sum of variables subject to the bound on their product by setting them all equal to each other (Case 3).
However, when the individual variable constraints make this solution infeasible, those become active and the free variables must be adjusted (Cases 1 and 2).

\begin{lemma}
\label{lem:opt}
Consider the following optimization problem:
$$\min_{\vc{x} \in \mathbb{R}^3} \x+\y+\z$$
such that 
$$\left(\frac{mnk}{P}\right)^2 \leq \x\y\z$$
$$\frac{nk}{P} \leq \x$$
$$\frac{mk}{P} \leq \y$$
$$\frac{mn}{P} \leq \z,$$
where $m\geq n \geq k \geq 1$ and $P\geq 1 $.
The optimal solution $\vc{x}^*$ depends on the relative values of the constraints, yielding three cases:
\begin{enumerate}
	\item if $P \leq \frac mn$, then $\x^*=nk$, $\y^*=\frac{mk}{P}$, $\z^*=\frac{mn}{P}$;
	\item if $\frac mn \leq P \leq \frac{mn}{k^2}$, then $\x^*=\y^*=\left(\frac{mnk^2}{P}\right)^{1/2}$, $\z^*= \frac{mn}{P}$;
	\item if $\frac{mn}{k^2} \leq P$, then $\x^*=\y^*=\z^*= \big(\frac{mnk}{P}\big)^{\frac{2}{3}}$.
\end{enumerate}
This can be visualized as follows:
\begin{center}
	\begin{tikzpicture}[scale=.95, every node/.style={transform shape}]
	\draw [->, thick] (-0.1,0) -- (15,0) node [below right,pastelgreen,scale=1] {$P$};
	\draw (0, 0.1) -- node [below, pastelred, scale=1]{$1$}(0,-0.1);
	\draw (5, 0.1) -- node [below, pastelred, scale=1]{$\frac{m}{n}$}(5,-0.1);
	\draw (10, 0.1) -- node [below, pastelred, scale=1] {$\frac{mn}{k^2}$}(10,-0.1);
	
	\node[align=left,below,scale=1] at (2.5, -0.4) {$\x^*=nk$\\ $\y^*=\frac{mk}{P}$\\ $\z^*=\frac{mn}{P}$};
	\node[align=left,below,scale=1] at (7.5, -0.6) {$\x^*=\y^*=\big(\frac{mnk^2}{P}\big)^{1/2}$\\$\z^*=\frac{mn}{P}$};
	\node[align=center,below,scale=1] at (12.75, -0.8) {$\x^*=\y^*=\z^*= \big(\frac{mnk}{P}\big)^{2/3}$};	
	\end{tikzpicture}
\end{center} 
\end{lemma}

\begin{proof}
By \cref{lem:KKT}, we can establish the optimality of the solution for each case by verifying that there exist dual variables such that the KKT conditions specified in \cref{def:KKT} are satisfied.
This optimization problem fits the assumptions of \cref{lem:KKT} because the objective function and all but the first constraint are affine functions, which are convex and quasiconvex, and the first constraint is quasiconvex on the positive octant (which contains the intersection of the affine constraints) by \cref{lem:quasiconvex}.

To match standard notation (and that of \cref{lem:KKT}), we let 
\begin{equation*}
f(\vc{x}) = \x+\y+\z 
\end{equation*}
and
\begin{equation*}
\vc{g}(\vc{x}) = 
\begin{bmatrix} 
(mnk/P)^2-\x\y\z \\
nk/P - \x \\
mk/P - \y \\
mn/P - \z \\
\end{bmatrix}.
\end{equation*}
Thus the gradient of the objective function is $\nabla f(\vc{x}) = \begin{bmatrix} 1 & 1 & 1 \end{bmatrix} $ and the Jacobian of the constraint function is
\begin{equation*}
\vc{J}_{\vc{g}}(\vc{x}) = \begin{bmatrix}
-x_2x_3 & -x_1x_3 & -x_1x_2 \\
-1 & 0 & 0 \\
0 & -1 & 0 \\
0 & 0 & -1 \\
\end{bmatrix}.
\end{equation*}

\paragraph{Case 1 ($P \leq \frac nm$)} 
We let 
$$\vc{x}^* = \begin{bmatrix} nk & \frac{mk}{P} & \frac{mn}{P} \end{bmatrix}$$ 
and 
$$\vc{\mu}^* = \begin{bmatrix} \frac{P^2}{m^2nk} & 0 & 1 - \frac{Pn}{m} & 1-\frac{Pk}{m} \end{bmatrix}$$
and verify the KKT conditions.
Primal feasibility is immediate, and dual feasibility follows from $P \leq \frac mn \leq \frac mk$, the condition of this case and by the assumption $n\geq k$.
Stationarity follows from direct verification that 
$$\vc{\mu}^* \cdot \vc{J}_{\vc{g}}(\vc{x}^*) = \begin{bmatrix} -1 & -1 & -1 \end{bmatrix}.$$
Complementary slackness is satisfied because the only nonzero dual variables are $\mu_1^*$, $\mu_3^*$, and $\mu_4^*$, and the 1st, 3rd, and 4th constraints are tight.

\paragraph{Case 2 ($\frac mn \leq P \leq \frac{mn}{k^2}$)} 
We let 
$$\vc{x}^* = \begin{bmatrix} \left(\frac{mnk^2}{P}\right)^{1/2} & \left(\frac{mnk^2}{P}\right)^{1/2} & \frac{mn}{P} \end{bmatrix}$$ 
and 
$$\vc{\mu}^* = \begin{bmatrix} \left(\frac{P}{mnk^{2/3}}\right)^{3/2} & 0 & 0 & 1-\left(\frac{Pk^2}{mn}\right)^{1/2} \end{bmatrix}$$
and verify the KKT conditions.
The primal feasibility of $x_1=x_2$ is satisfied because 
$$\frac{nk}{P} \leq \frac{mk}{P} \leq \left(\frac{mnk^2}{P}\right)^{1/2}$$
where the first inequality follows from the assumption $m\geq n$ and the second inequality follows from $m/n \leq P$ (one condition of this case).
The other constraints are clearly satisfied.
Dual feasibility requires that $1-(Pk^2/mn)^{1/2} \geq 0$, which is satisfied because $P \leq mn/k^2$ (the other condition of this case).
Stationarity can be directly verified.
Complementary slackness is satisfied because the 1st and 4th constraints are both tight for $\vc{x}^*$, corresponding to the only nonzeros in $\vc{\mu}^*$.

\paragraph{Case 3 ($\frac{mn}{k^2} \leq P$)} 
We let 
$$\vc{x}^* = \begin{bmatrix} \left(\frac{mnk}{P}\right)^{2/3} & \left(\frac{mnk}{P}\right)^{2/3} & \left(\frac{mnk}{P}\right)^{2/3} \end{bmatrix}$$ 
and 
$$\vc{\mu}^* = \begin{bmatrix} \left(\frac{P}{mnk}\right)^{4/3} & 0 & 0 & 0 \end{bmatrix}$$
and verify the KKT conditions.
We first consider the primal feasibility conditions.
We have
\begin{equation*}
\frac{nk}{P} \leq \frac{mk}{P} \leq \frac{mn}{P} \leq \left(\frac{mnk}{P}\right)^{2/3},
\end{equation*}
where the first two inequalities are implied by the assumption $m\geq n \geq k$ and the last follows from $\frac{mn}{k^2} \leq P$, the condition of this case.
Dual feasibility is immediate, and stationarity is directly verified.
Complementary slackness is satisfied because the 1st constraint is tight for $\vc{x}^*$ and $\mu_1^*$ is the only nonzero.

Note that the optimal solutions coincide at boundary points between cases so that the values are continuous as $P$ varies.
\end{proof}

\subsection{Communication Lower Bound}

We now state our main result, memory-independent communication lower bounds for general matrix multiplication with tight constants.
After the general result, we also present a corollary for square matrix multiplication.
The tightness of the constants in the lower bound is proved in \cref{sec:algo}.

\begin{theorem}
\label{thm:main}
Consider a classical matrix multiplication computation involving matrices of size $n_1\times n_2$ and $n_2\times n_3$.
Let $m=\max\{n_1,n_2,n_3\}$, $n=\emph{\text{median}} \{n_1,n_2,n_3\}$, and $k = \min\{n_1,n_2,n_3\}$, so that $m\geq n\geq k$.  
Any parallel algorithm using $P$ processors that starts with one copy of the two input matrices and ends with one copy of the output matrix and load balances either the computation or the data must communicate at least
$$D - \frac{mn+mk+nk}{P} \text{ words of data},$$
where
$$D = 
\begin{cases} 
	\quad \frac{mn+mk}{P} + nk & \text{ if } \quad 1\leq P \leq \frac mn \\
	\quad 2\left(\frac{mnk^2}{P}\right)^{1/2} + \frac{mn}{P} & \text{ if } \quad \frac mn \leq P \leq \frac{mn}{k^2} \\
	\quad 3\left(\frac{mnk}{P}\right)^{2/3} & \text{ if } \quad \frac{mn}{k^2} \leq P. 
\end{cases}$$
\end{theorem}

\begin{proof}
To establish the lower bound, we focus on a single processor.
If the algorithm load balances the computation, then every processor performs $mnk/P$ scalar multiplications, and there exists some processor whose input data at the start of the algorithm plus output data at the end of the algorithm must be at most $(mn+mk+nk)/P$ words of data (otherwise the algorithm would either start with more than one copy of the input matrices or end with more than one copy of the output matrix).
If the algorithm load balances the data, then every processor starts and end with a total of $(mn+mk+nk)/P$ words, and some processor must perform at least $mnk/P$ scalar multiplications (otherwise fewer than $mnk$ multiplications are performed).
In either case, there exists a processor that performs at least $mnk/P$ multiplications and has access to at most $(mn+mk+nk)/P$ data.

Let $F$ be the set of multiplications assigned to this processor, so that $|F|\geq mnk/P$.
Each element of $F$ can be indexed by three indices $(i_1,i_2,i_3)$ and corresponds to the multiplication of $\A(i_1,i_2)$ with $\B(i_2,i_3)$, which contributes to the result $\CC(i_1,i_3)$.
Let $\phi_{\A}(F)$ be the projection of the set $F$ onto the matrix $\A$, so that $\phi_{\A}(F)$ are the entries of $\A$ required for the processor to perform the scalar multiplications in $F$.
Here, elements of $\phi_{\A}(F)$ can be indexed by two indices: $\phi_{\A}(F) = \{(i_1,i_2): \exists \; i_3 \text{ s.t. } (i_1,i_2,i_3) \in F\}$.
Define $\phi_{\B}(F)$ and $\phi_{\CC}(F)$ similarly.
The processor must access all of the elements in $\phi_{\A}(F)$, $\phi_{\B}(F)$, and $\phi_{\CC}(F)$ in order to perform all the scalar multiplications in $F$.
Because the processor starts and ends with at most $(mn+mk+nk)/P$ data, the communication performed by the processor is at least
\begin{equation*}
|\phi_{\A}(F)|+|\phi_{\B}(F)|+|\phi_{\CC}(F)| - \frac{mn+mk+nk}{P},
\end{equation*}
which is a lower bound on the communication along the critical path of the algorithm.

In order to lower bound $|\phi_{\A}(F)|+|\phi_{\B}(F)|+|\phi_{\CC}(F)|$, we form a constrained minimization problem with this expression as the objective function and constraints derived from  \cref{lem:LW,lem:projlb}.
The Loomis-Whitney inequality (\cref{lem:LW}) implies that
\begin{equation*}
|\phi_{\A}(F)|\cdot |\phi_{\B}(F)|\cdot |\phi_{\CC}(F)| \geq |F| \geq \frac{n_1n_2n_3}{P}=\frac{mnk}{P},
\end{equation*}
and the lower bound on the projections from \cref{lem:projlb} means
\begin{equation*}
|\phi_{\A}(F)| \geq \frac{n_1n_2}{P}, \quad |\phi_{\B}(F)| \geq \frac{n_2n_3}{P}, \quad |\phi_{\CC}(F)| \geq \frac{n_1n_3}{P}.
\end{equation*}
For any algorithm, the processor's projections must satisfy these constraints, so the sum of their sizes must be at least the minimum value of optimization problem.
Then by \cref{lem:opt} (and assigning the projections to $\x,\y,\z$ appropriately based on the relative sizes of $n_1,n_2,n_3$), the result follows.
\end{proof}

We also state the result for square matrix multiplication, which is a direct application of \cref{thm:main} with $n_1=n_2=n_3$.

\begin{corollary}
Consider a classical matrix multiplication computation involving two matrices of size $n\times n$.
Any parallel algorithm using $P$ processors that starts with one copy of the input data and ends with one copy of the output data and load balances either the computation or the data must communicate at least
$$3\frac{n^2}{P^{2/3}} - 3\frac{n^2}{P} \text{ words of data}.$$
\end{corollary}

\section{Optimal Parallel Algorithm}
\label{sec:algo}

In this section we present an optimal parallel algorithm (\cref{alg:3dmatmul}) to show that the lower bound (including the constants) is tight.
The idea is to organize the processors into a 3D processor grid and assign the computation of the matrix multiplication (a 3D iteration space) to processors according to their location in the grid.
The algorithm is not new, but we present it here in full detail for completeness and to provide the complete analysis, which has not appeared before.
In particular, \cref{alg:3dmatmul} is nearly identical to the one proposed by Aggarwal et al.~\cite{ABGJP95}, though they use the LPRAM model and analyze only the case where $P$ is large.
In the LPRAM model, for example, processors can read concurrently from a global shared memory, while in the $\alpha$-$\beta$-$\gamma$ model, the data is distributed across local memories and is shared via collectives like All-Gathers.
Demmel et al.~\cite{DE+13} present and analyze their recursive algorithm to show its asymptotic optimality in all three cases, but they do not track constants.
See \cref{sec:relatedalgs} for more discussion of previous work on optimal parallel algorithms.

Consider the multiplication of an $n_1\times n_2$ matrix $\A$ with an $n_2\times n_3$ matrix $\B$, and let $\CC=\A\cdot\B$.
\Cref{alg:3dmatmul} organizes the $P$ processors into a $3$-dimensional $p_1 \times p_2 \times p_3$ logical processor grid with $p_1p_2p_3=P$.
Note that one or two of the processor grid dimensions may be equal to 1, which simplifies to a 2- or 1-dimensional grid.
A processor coordinate is represented as $(p_1^\prime, p_2^\prime, p_3^\prime)$, where $1\le p_{k}^\prime \le p_k$, for $k=1,2,3$.

The basic idea of the algorithm is to perform two collective operations, All-Gathers, so that each processor receives the input data it needs to perform all of its computation (in an All-Gather, all the processors involved end up with the union of the input data that starts on each processor).
The result of each local computation must be summed with all other contributions to the same output matrix entries from other processors, and the summations are performed via a Reduce-Scatter collective operation (in a Reduce-Scatter, the sum of the input data from all processors is computed so that it ends up evenly distributed across processors).

\begin{algorithm}
	\caption{\label{alg:3dmatmul}Comm-Optimal Parallel Matrix Multiplication}
	\begin{algorithmic}[1]
		\REQUIRE $\A$ is $n_1\times n_2$, $\B$ is $n_2\times n_3$, $p_1 \times p_2 \times p_3$ logical processor grid
		\ENSURE $\CC = \A\cdot\B$ is $n_1 \times n_3$
		\STATE $(p_1^\prime, p_2^\prime, p_3^\prime)$ is my processor ID
		\STATE // Gather input matrix data
		\STATE $\A_{p_1^\prime p_2^\prime}$ = All-Gather($\A_{p_1^\prime p_2^\prime p_3^\prime}$, $(p_1^\prime, p_2^\prime, :)$)\label{alg:3dmatmul:line:allGatherMatrixA}
		\STATE $\B_{p_2^\prime p_3^\prime}$ = All-Gather($\B_{p_1^\prime p_2^\prime p_3^\prime}$, $(:, p_2^\prime, p_3^\prime)$)\label{alg:3dmatmul:line:allGatherMatrixB}
		\STATE // Perform local computation 
		\STATE $\D_{p_1^\prime p_2^\prime p_3^\prime}=\A_{p_1^\prime p_2^\prime} \cdot \B_{p_2^\prime p_3^\prime}$\label{alg:3dmatmul:line:localComputation}
		\STATE // Sum results to compute $\CC_{p_1^\prime p_3^\prime}$
		\STATE $\CC_{p_1^\prime p_2^\prime p_3^\prime}$ = Reduce-Scatter($\D_{p_1^\prime p_2^\prime p_3^\prime}$, $(p_1^\prime, :, p_3^\prime)$)\label{alg:3dmatmul:line:reduceScatterC}
	\end{algorithmic}
\end{algorithm}

\begin{figure}[ht!]
\centering
\begin{tikzpicture}[every node/.append style={transform shape},scale=1.65]
\draw[line width=3,stealth-stealth,red!75] (.1,2,0) -- (.1,2,-2.25);
\draw[line width=3,stealth-stealth,red!75] (0,2,0) -- (-2.25,2,0);
\draw[line width=3,stealth-stealth,red!75] (0,2,0) -- (0,-.25,0);
\begin{scope}[canvas is yz plane at x=.5,rotate=-90,yscale=-1,shift={(-.5,-3+.5)}]
	\draw[fill=red!25] (0,0) rectangle (1,1);
	\draw[fill=red!75] (2/3,0) rectangle (1,1);
	\draw[black] (0,0) grid (3,3);
	\draw[black,xscale=1/3,dotted] (0,0) grid (3,1);
	\node[yscale=-1,scale=2] at (3/2,3/2) {\Large $\CC$};
	\node[yscale=-1,scale=.5] at (1/2,1/2) {$\CC_{11}$};
	\node[yscale=-1,scale=.5] at (1/2,3/2) {$\CC_{12}$};
	\node[yscale=-1,scale=.5] at (1/2,5/2) {$\CC_{13}$};
	\node[yscale=-1,scale=.5] at (3/2,1/2) {$\CC_{21}$};
	\node[yscale=-1,scale=.5] at (3/2,5/2) {$\CC_{23}$};
	\node[yscale=-1,scale=.5] at (5/2,1/2) {$\CC_{31}$};
	\node[yscale=-1,scale=.5] at (5/2,3/2) {$\CC_{32}$};
	\node[yscale=-1,scale=.5] at (5/2,5/2) {$\CC_{33}$};
\end{scope}
\begin{scope}[canvas is yx plane at z=.5,yscale=-1,rotate=180,shift={(-3+.5,-3+.5)}]
	\draw[fill=red!25] (0,2) rectangle (1,3);
	\draw[fill=red!75] (0,2) rectangle (1,7/3);
	\draw[black] (0,0) grid (3,3);
	\draw[black,shift={(0,2)},yscale=1/3,dotted] (0,0) grid (1,3);
	\node[rotate=90,scale=2] at (3/2,3/2) {\Large $\A$};
	\node[rotate=90,scale=.5] at (1/2,1/2) {$\A_{11}$};
	\node[rotate=90,scale=.5] at (1/2,3/2) {$\A_{12}$};
	\node[rotate=90,scale=.5] at (1/2,5/2) {$\A_{13}$};
	\node[rotate=90,scale=.5] at (3/2,1/2) {$\A_{21}$};
	\node[rotate=90,scale=.5] at (3/2,5/2) {$\A_{23}$};
	\node[rotate=90,scale=.5] at (5/2,1/2) {$\A_{31}$};
	\node[rotate=90,scale=.5] at (5/2,3/2) {$\A_{32}$};
	\node[rotate=90,scale=.5] at (5/2,5/2) {$\A_{33}$};
\end{scope}
\begin{scope}[canvas is zx plane at y=(3-.5),rotate=90,shift={(-3+.5,-.5)}]
	\draw[fill=red!25] (2,0) rectangle (3,1);
	\draw[fill=red!75] (2,0) rectangle (3,1/3);
	\draw[black] (0,0) grid (3,3);
	\draw[black,shift={(2,0)},yscale=1/3,dotted] (0,0) grid (1,3);
	\node[rotate=90,scale=2] at (3/2,3/2) {\Large $\B$};
	\node[rotate=90,scale=.5] at (1/2,1/2) {$\B_{11}$};
	\node[rotate=90,scale=.5] at (1/2,3/2) {$\B_{12}$};
	\node[rotate=90,scale=.5] at (1/2,5/2) {$\B_{13}$};
	\node[rotate=90,scale=.5] at (3/2,1/2) {$\B_{21}$};
	\node[rotate=90,scale=.5] at (3/2,5/2) {$\B_{23}$};
	\node[rotate=90,scale=.5] at (5/2,1/2) {$\B_{31}$};
	\node[rotate=90,scale=.5] at (5/2,3/2) {$\B_{32}$};
	\node[rotate=90,scale=.5] at (5/2,5/2) {$\B_{33}$};
\end{scope}
\end{tikzpicture}
\caption{Visualization of \cref{alg:3dmatmul} with a $3\times 3\times 3$ processor grid.  The 3D iteration space is mapped onto the processor grid, and the matrices are mapped to the faces of the grid.  The dark highlighting corresponds to the input data initially owned and the output data finally owned by processor $(1,3,1)$, and the light highlighting signifies the data of other processors it uses to perform the local computation.  The arrows show the sets of processors involved in the three collective operations involving processor $(1,3,1)$.}
\label{fig:3dmatmul}
\end{figure}
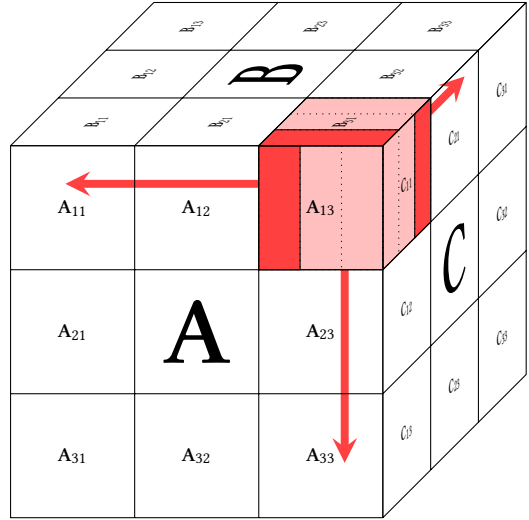

\Cref{alg:3dmatmul} imposes requirements on the initial distribution of the input matrices and the final distribution of the output.
These conditions do not always hold in practice, but to show that the lower bound (which makes no assumption on data distribution except that only 1 copy of the input exists at the start of the computation) is tight, we allow the algorithm to specify its distributions.
For simplicity of explanation, we assume that $p_1$, $p_2$, $p_3$ evenly divide $n_1$, $n_2$, $n_3$, respectively. 
We use the notation $\A_{p_1^\prime p_2^\prime}$ to denote the submatrix of $\A$ such that
$$\A_{p_1^\prime p_2^\prime} = \A\left((p_1^\prime-1) \cdot \frac{n_1}{p_1} + 1 : p_1^\prime \cdot \frac{n_1}{p_1},  (p_2^\prime-1) \cdot \frac{n_2}{p_2} + 1 : p_2^\prime \cdot \frac{n_2}{p_2}\right),$$
and we define $\B_{p_2^\prime p_3^\prime}$ and $\CC_{p_1^\prime p_3^\prime}$ similarly.
The algorithm assumes that at the start of the computation, $\A_{p_1^\prime p_2^\prime}$ is distributed evenly across processors $(p_1^\prime, p_2^\prime, :)$ and $\B_{p_2^\prime p_3^\prime}$ is distributed evenly across processors $(:, p_2^\prime, p_3^\prime)$.
We define $\A_{p_1^\prime p_2^\prime p_3^\prime}$ and $\B_{p_1^\prime p_2^\prime p_3^\prime}$ as the elements of the input matrices that processor $(p_1^\prime, p_2^\prime, p_3^\prime)$ initially owns.
At the end of the algorithm, $\CC_{p_1^\prime p_3^\prime}$ is distributed evenly across processors $(p_1^\prime, :, p_3^\prime)$, and we let $\CC_{p_1^\prime p_2^\prime p_3^\prime}$ be the elements owned by processor $(p_1^\prime, p_2^\prime, p_3^\prime)$.

\Cref{fig:3dmatmul} presents a visualization of \cref{alg:3dmatmul}.
In this example, we have $n_1=n_2=n_3$, and 27 processors are arranged in a $3\times 3\times 3$ grid.
We highlight the data and communication of a particular processor with ID $(1,3,1)$.
The dark highlighting corresponds to the input data initially owned by the processor ($\A_{131}$ and $\B_{131}$) as well as the output data owned by the processor at the end of the computation ($\CC_{131}$).
The figure shows each of these submatrices as block columns of the submatrices $\A_{13}$, $\B_{31}$, and $\CC_{11}$, but any even distribution of these across the same set of processors suffices.
The light highlighting of the submatrices $\A_{13}$, $\B_{31}$, and $\CC_{11}$ corresponds to the data of other processors involved in the local computation on processor $(1,3,1)$, and their size corresponds to the communication cost.
The three collectives that involve processor $(1,3,1)$ occur across three different fibers in the processor grid, as depicted by the arrows in the figure.

\subsection{Cost Analysis}
\label{sec:cost_analysis}

Now we analyze computation and communication costs of the algorithm. 
Each processor performs $\frac{n_1}{p_1}.\frac{n_2}{p_2}.\frac{n_3}{p_3} = \frac{n_1n_2n_3}{P}$ local computations in Line~\ref{alg:3dmatmul:line:localComputation}. 
Communication occurs only in the All-Gather and Reduce-Scatter collectives in Lines~\ref{alg:3dmatmul:line:allGatherMatrixA}, \ref{alg:3dmatmul:line:allGatherMatrixB},  and~\ref{alg:3dmatmul:line:reduceScatterC}. 
Each processor is involved in two All-Gathers involving input matrices and one Reduce-Scatter involving the output matrix. 
Lines~\ref{alg:3dmatmul:line:allGatherMatrixA}, \ref{alg:3dmatmul:line:allGatherMatrixB} specify simultaneous All-Gathers across sets of $p_3$, $p_1$ processors, respectively, and Line~\ref{alg:3dmatmul:line:reduceScatterC} specifies simultaneous Reduce-Scatters across sets of $p_2$ processors. 
Note that if $p_k=1$ for any $k=1,2,3$, then the corresponding collective can be ignored as no communication occurs.
The difference between \cref{alg:3dmatmul} and \cite[Algorithm 1]{ABGJP95} is the Reduce-Scatter collective, which replaces the All-to-All collective and has smaller latency cost.

We assume that bandwidth-optimal algorithms, such as bidirectional exchange or recursive doubling/halving, are used for the All-Gather and Reduce-Scatter collectives. 
The optimal communication cost of these collectives on $p$ processors is $(1-\frac{1}{p})w$, where $w$ is the words of data in each processor after All-Gather or before Reduce-Scatter collective~\cite{Thakur:CollectiveCommunications:2005,Chan:CollectiveCommunications:2007}. 
Each processor also performs $(1-\frac{1}{p})w$ computations for the Reduce-Scatter collective.

\newcommand{\drawcube}[6]{
\pgfmathsetmacro{\mdim}{#1}
\pgfmathsetmacro{\ndim}{#2}
\pgfmathsetmacro{\kdim}{#3}
\pgfmathsetmacro{\pdim}{#4}
\pgfmathsetmacro{\qdim}{#5}
\pgfmathsetmacro{\rdim}{#6}
\pgfmathsetmacro{\offset}{.5}

\node[shift={(.5,-.5,.5)},scale=2] at (-\kdim/2,-\offset,0) {$n$};
\node[shift={(.5,-.5,.5)},scale=2] at (\offset/2,-\offset/2,-\ndim/2) {$k$};
\node[shift={(.5,-.5,.5)},scale=2] at (-\kdim-\offset,\mdim/2,0) {$m$};

\begin{scope}[canvas is yz plane at x=.5,rotate=-90,yscale=-1,shift={(-.5,-\mdim+.5)}]
	\draw[black,xscale=\ndim/\qdim,yscale=\mdim/\pdim] (0,0) grid (\qdim,\pdim);
	\node[yscale=-1,scale=2] at (\ndim/2,\mdim/2) {\Large $\CC$};
\end{scope}
\begin{scope}[canvas is yx plane at z=.5,yscale=-1,rotate=180,shift={(-\mdim+.5,-\kdim+.5)}]
	\draw[black,xscale=\mdim/\pdim,yscale=\kdim/\rdim] (0,0) grid (\pdim,\rdim);
	\node[rotate=90,scale=2] at (\mdim/2,\kdim/2) {\Large $\A$};
\end{scope}
\begin{scope}[canvas is zx plane at y=(\mdim-.5),rotate=90,shift={(-\kdim+.5,-.5)}]
	\draw[black,xscale=\kdim/\rdim,yscale=\ndim/\qdim] (0,0) grid (\rdim,\qdim);
	\node[rotate=90,scale=2] at (\kdim/2,\ndim/2) {\Large $\B$};
\end{scope}
}

\begin{figure*}[ht]
\begin{subfigure}{.3\linewidth}
	\centering
	\begin{tikzpicture}[every node/.append style={transform shape},scale=.6]
		\drawcube{16}{1}{4}{3}{1}{1}
	\end{tikzpicture}
	\caption{1D case: $P=3$ with grid $3\times 1\times 1$}
	\label{fig:MM:1D}
\end{subfigure} 
\begin{subfigure}{.3\linewidth}
	\centering
	\begin{tikzpicture}[every node/.append style={transform shape},scale=.6]
		\drawcube{16}{1}{4}{12}{1}{3}
	\end{tikzpicture}
	\caption{2D case: $P=36$ with grid $12 \times 3 \times 1$}
	\label{fig:MM:2D}
\end{subfigure} 
\begin{subfigure}{.3\linewidth}
	\centering
	\begin{tikzpicture}[every node/.append style={transform shape},scale=.6]
		\drawcube{16}{1}{4}{32}{2}{8}
	\end{tikzpicture}
	\caption{3D case: $P=512$ with grid $32 \times 8 \times 2$}
	\label{fig:MM:3D}
\end{subfigure}
\caption{Example parallelizations of iteration space of multiplication of a $9600 \times 2400$ matrix $\A$ with a $2400 \times 600$ matrix $\B$}
\label{fig:MM:123D}
\end{figure*}
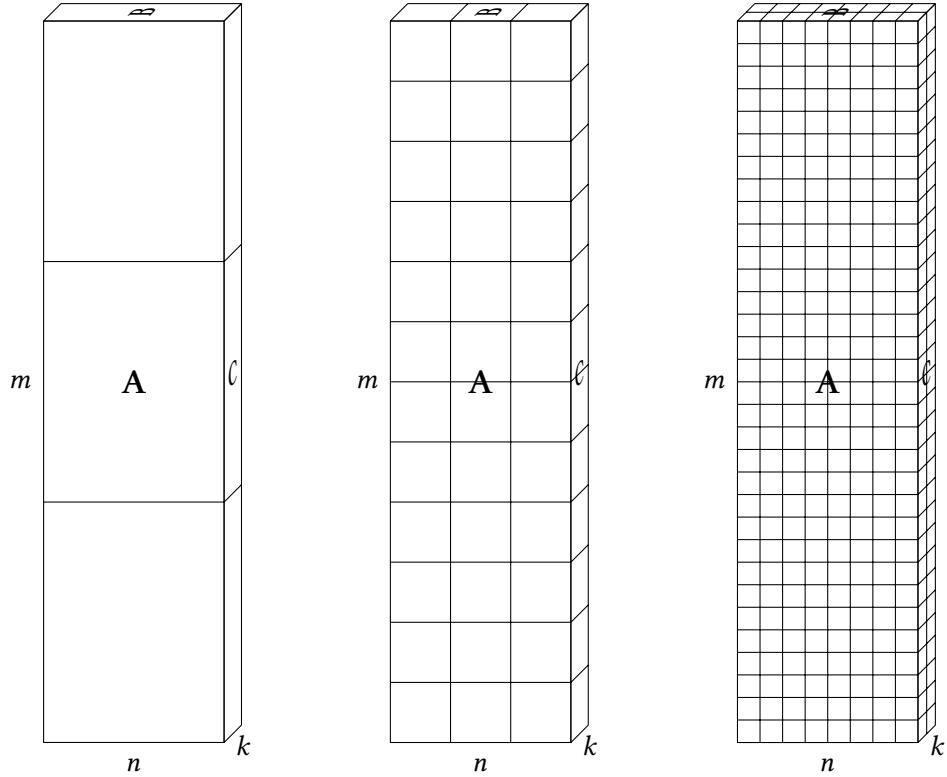

Hence the communication costs of Lines~\ref{alg:3dmatmul:line:allGatherMatrixA}, \ref{alg:3dmatmul:line:allGatherMatrixB} in \Cref{alg:3dmatmul} are $(1-\frac{1}{p_3}) \frac{n_1n_2}{p_1p_2}$ and $(1-\frac{1}{p_1}) \frac{n_2n_3}{p_2p_3}$, respectively, to accomplish All-Gather operations, and the communication cost of performing the Reduce-Scatter operation in Line~\ref{alg:3dmatmul:line:reduceScatterC} is $(1-\frac{1}{p_2}) \frac{n_1n_3}{p_1p_3}$. 
Note that if $p_k=1$ for any $k=1,2,3$, then the cost of the corresponding collective reduces to 0.
Thus the overall cost of \Cref{alg:3dmatmul} is 
\begin{equation}
\label{eq:cost}
\frac{n_1n_2}{p_1p_2} + \frac{n_2n_3}{p_2p_3} + \frac{n_1n_3}{p_1p_3} - \left(\frac{n_1n_2 + n_2n_3 +n_1n_3}{P}\right).
\end{equation} 
Due to Reduce-Scatter operation, each processor also performs $(1-\frac{1}{p_2}) \frac{n_1n_3}{p_1p_3}$ computations, which is dominated by the $\frac{n_1n_2n_3}{P}$ computations of Line~\ref{alg:3dmatmul:line:localComputation}.

\subsection{Optimal Processor Grid Selection}

The communication cost of \Cref{alg:3dmatmul}, given by \cref{eq:cost}, depends on the processor grid dimensions.
Here we discuss how to select the processor grid dimensions such that the lower bound on communication given in~\cref{thm:main} is attained by~\cref{alg:3dmatmul} given the matrices dimensions $n_1,n_2$ and $n_3$ and the number of processors $P$. 
As before, we let $m,n,k$ represent the maximum, median, and minimum values of the three dimensions.
Letting $p_1,p_2,p_3$ be the grid dimensions, we similarly define $p,q,r$ to be the processor grid dimensions corresponding to matrix dimensions $m,n,k$, respectively.
Because the order of processor grid dimensions will be chosen to be consistent with the matrix dimensions, we will have $p\geq q\geq r$.
To demonstrate the tightness of the lower bound, the analysis below assumes that the processor grid dimensions divide the matrices dimensions. 

Following~\cref{thm:main}, depending on the relative sizes of the aspect ratios among matrix dimensions and the number of processors, we encounter three cases that correspond to 3D, 2D, and 1D processor grids.
That is, when $p_i=1$ for one value of $i$, then the processor grid is effectively 2D, and when $p_i=1$ for two values of $i$, the processor grid is effectively 1D.
In the following we show how to obtain the grid dimensions and show that \Cref{alg:3dmatmul} attains the communication lower bound given in~\cref{thm:main} in each case.

First, suppose $P \le \frac{m}{n}$.
In this case, we set $r=q=1$, and set $p=P$ to obtain a 1D grid.
From \cref{eq:cost}, \Cref{alg:3dmatmul} has communication cost
\begin{equation*}
\frac{mn+mk}{P} + nk- \frac{mn + mk + nk}{P} = \left(1-\frac1P\right) nk,
\end{equation*}
which matches the 1st case of \cref{thm:main}.

Now suppose $\frac{m}{n} < P \le \frac{mn}{k^2}$.
We set $r=1$, and set $p$ and $q$ such that $\frac{m}{p} = \frac{n}{q}$, yielding $p = \left(\frac{P}{mn}\right)^{1/2} m$ and $q =\left(\frac{P}{mn}\right)^{1/2}n$.
Note that the assumption $\frac mn < P$ is required so that $q> 1$, and $p> 1$ also follows.
Our analysis also assumes that $p$ and $q$ are integers, which is sufficient to show that the lower bound is tight in general as there are an infinite number of dimensions for which the assumption holds.
In this case, we have a 2D processor grid, and \Cref{alg:3dmatmul} has communication cost
\begin{equation*}
\frac{mn}{pq}+\frac{mk}{p}+\frac{nk}{q} - \frac{mn + mk + nk}{pq} = 2\left(\frac{mnk^2}{P}\right)^{1/2} - \frac{mk + nk}{P},
\end{equation*}
matching the 2nd case of \cref{thm:main}.

Finally, suppose $\frac{mn}{k^2} < P$. 
As suggested in \cite{ABGJP95}, we set the grid dimensions such that $\frac{m}{p} = \frac{n}{q} = \frac{k}{r}$. 
That is, $r=\left(\frac{P}{mnk}\right)^{1/3}k$, $q=\left(\frac{P}{mnk}\right)^{1/3}n$, and $p=\left(\frac{P}{mnk}\right)^{1/3}m$.
Note that the assumption $\frac{mn}{k^2} < P$ is required so that $r> 1$ (which also implies $q> 1$ and $p> 1$).
This assumption was implicit in the analysis of \cite{ABGJP95}.
Again, we assume that $p,q,r$ are integers.
Thus, we have a 3D processor grid and a communication cost of
\begin{equation*}
3\left(\frac{mnk}{P}\right)^{2/3} -  \frac{mn + mk + nk}{P},
\end{equation*}
which matches the 3rd case of \cref{thm:main}.

Comparing the obtained communication cost in each case with the lower bound obtained in~\cref{thm:main} we conclude that \Cref{alg:3dmatmul} is optimal given the grid dimensions are selected as above.

\subsection{Optimal Processor Grid Examples}

\Cref{fig:MM:123D} illustrates each of the three cases for a fixed set of matrix dimensions.
Here we consider multiplying a $9600 \times 2400$ matrix $\A$ with a $2400 \times 600$ $\B$ so that $\CC$ is $9600 \times 600$, so in our notation with $m\geq n\geq k$, $\A$ is $m\times n$, $\B$ is $n\times k$, and $\CC$ is $m\times k$.
The 3D $m\times n\times k$ iteration space is visualized with faces corresponding to correctly oriented matrices.
In this example, we consider $P\in\{3, 36, 512\}$.

With 3 processors, we fall into the 1st case, as $P \le \frac{m}{n}=4$, and the optimal processor grid is $3\times 1\times 1$, which is 1D as shown in \cref{fig:MM:1D}.
Note that the computation assigned to each processor is not a cube in this case, that is, $\frac mp \neq \frac nq \neq \frac kr$.
The only data that must be communicated are entries of $\B$, though all processors need all of $\B$.

When $P=36$, we fall into the 2nd case, and the optimal processor grid is 2D: $12\times3\times1$, as shown in \cref{fig:MM:2D}.
Here we see that the iteration space assigned to each processor is $800\times 800\times 600$, so $\frac mp = \frac nq \neq \frac kr$.
In this case, entries of $\B$ and $\CC$ must be communicated, but each entry of $\A$ is required by only one processor.

Finally, for $P=512$, we satisfy $P> \frac{mn}{k^2}=64$ and fall into the 3rd case.
The optimal processor grid is shown in \cref{fig:MM:3D} to be $32\times8\times2$, and we see that the local computation of each processor is a cube: $\frac mp = \frac nq = \frac kr$.
For a 3D grid, entries of all 3 matrices are communicated.

\section{Conclusion}
\label{sec:conclusion}

\Cref{thm:main} establishes memory-independent communication lower bounds for parallel matrix multiplication.
By casting the lower bound on accessed data as the solution to a constrained optimization problem, we are able to obtain a result with explicit constants spanning over three scenarios that depend on the relative sizes of the matrix aspect ratios and the number of processors.
\Cref{alg:3dmatmul} demonstrates that the constants established in \cref{thm:main} are tight, as the algorithm is general enough to be applied in each of the three scenarios by tuning the processor grid.
As we discuss below, our lower bound proof technique tightens the constants proved in earlier work, and we believe it can be generalized to improve known communication lower bounds for other computations.

\subsection{Comparison to Existing Results}
\label{sec:comparison}

We now provide full details of the constants presented in \cref{tab:summary}, and compare the previous results with the constants of \cref{thm:main}.
The first row of the table gives the constant from the proof by Aggarwal, Chandra, and Snir \cite[Theorem 2.3]{ACS90}.
While the result is stated asymptotically, an explicit constant is given in a key lemma (\cite[Lemma 2.2]{ACS90}) used in the proof, from which we can derive the constant for the main result.

The second row of the table corresponds to the work of Irony, Toledo, and Tiskin \cite{ITT04}, who establish the first parallel bounds for matrix multiplication.
Their memory-independent  bound is stated for square matrices with a parametrized prefactor corresponding to the amount of local memory available \cite[Theorem 5.1]{ITT04}.
If we generalize it straightforwardly to rectangular dimensions and minimize the prefactor over any amount of local memory, then we obtain a bound of at least $1/2\cdot (mnk/P)^{2/3}$, which is asymptotically tight for $mn/k^2 \leq P$.
They do not provide any tighter results for $P < mn/k^2$.

The third row of the table corresponds to the results of Demmel et al.~\cite{DE+13}.
This work was the first to establish bounds for small values of $P$ and identify the three cases of asymptotic expressions.
\Cref{thm:main} obtains the same cases and leading order terms (up to constant factors) \cite[Table I]{DE+13}, and we present the explicit constant factors for leading order terms derived in \cite[Section II.B]{DE+13}.
We note that the boundaries between cases differ by a constant in that paper, which we do not reflect in \cref{tab:summary}.
Compared to these results, \cref{thm:main} establishes a tighter constant in all three cases.

We note that Kwasniewski et al. claim a combined result of memory-dependent and memory-independent parallel bounds \cite[Theorem 2]{KK+19}.
The memory-independent term has a constant that matches the 3rd case of \cref{thm:main}.
However, the proof includes a restrictive assumption on parallelization strategies, requiring that each processor is assigned a set of domains that are subblocks of the iteration space with dimensions $a\times a\times b$ for some $a,b$, and therefore does not apply to all parallel algorithms.

\subsection{Limited-Memory Scenarios}
\label{sec:mm:mdb}

The local memory required by \cref{alg:3dmatmul} matches the amount of communication performed plus the data already owned by the processor, which is given by the positive terms in \cref{eq:cost} and matches the value of $D$ in \cref{thm:main} with the optimal processor grid.
Note that the local memory $M$ must be large enough to store the inputs and output matrices, so $M \geq (mn+mk+nk)/P$.
When 1D or 2D processor grids are used, the local memory required is no more than a constant more than the minimum required to store the problem.
Further, \cref{alg:3dmatmul} can be adapted to reduce the temporary memory required to a negligible amount at the expense of higher latency cost but without affecting the bandwidth cost, and thus the algorithmic approach can be used even in extremely limited memory scenarios.
In the case of 3D processor grids, however, the temporary memory used by \cref{alg:3dmatmul} asymptotically dominates the minimum required, and thus the algorithm cannot be applied in limited-memory scenarios.
Reducing the memory footprint in this case necessarily increases the bandwidth cost.
Algorithms that smoothly trade off memory for communication savings in these limited memory scenarios are well studied \cite{MT99,SD11,BDHLS12-SS,KK+19}.

From the lower bound point of view, while \cref{thm:main} is always valid, it may not be the tightest bound in limited-memory scenarios.
The memory-dependent bound with leading term $2mnk/(P\sqrt M)$ (see \cite{SLLvdG19,KK+19,OLPSR20} and discussion in \cref{sec:related:seq}) can be larger.
In particular, this occurs when $mn/k^2 < P \leq 8/27 \cdot mnk/M^{3/2}$, and the memory-dependent bound dominates the memory-independent bound of $3(mnk/P)^{2/3}$ in that case.
This scenario implies that $M< 4/9 \cdot (mnk/P)^{2/3},$ in which case the temporary space required by \cref{alg:3dmatmul} exceeds the available memory.
Thus, the tightness of \cref{thm:main} for $mn/k^2 <P$ requires an assumption of sufficient memory.

When $P \leq mn/k^2$, the memory-independent bounds in the first two cases of \cref{thm:main} are always tight, with no assumption on local memory size.
That is, the memory-dependent bound never dominates the memory-independent bound.
Consider the 2nd case, so that $ m/n \leq P \leq mn/k^2$ and the memory-independent bound is $2(mnk^2/P)^{1/2}$.
Because the local memory must be large enough to store the largest matrix as well as the other two matrices, we have $M > mn/P$.
This implies $2mnk/(P\sqrt M) < 2 (mnk^2/P)^{1/2}$, so the memory-independent bound dominates.

Suppose further that $P \leq \frac mn$.
In this case, the leading-order term of the memory-independent bound is $nk$.
This 1st-case bound dominates the 2nd-case bound, which dominates the memory-dependent bound from the argument above.
Comparison of the full bounds of the 1st and 2nd cases simplifies to $2(mnk^2/P)^{1/2} \leq mk/P + nk$, which holds by the arithmetic-geometric mean inequality.

\subsection{Extensions}

The proof technique we use to obtain \cref{thm:main} is more generally applicable.
The basic approach of defining a constrained optimization problem to minimize the sum amount of data accessed subject to constraints on that data that depend on the nature of the computation has been used before for matrix multiplication \cite{SLLvdG19} and for tensor computations \cite{BKR18,BR20}.
The key to the results presented in this work is the imposition of lower bound constraints on the data accessed in each individual array given by \cref{lem:projlb}.
These lower bounds become active when the aspect ratios of the matrices are large relative to the number of processors and allow for tighter lower bounds in those cases.
The argument given in \cref{lem:projlb} is not specific to matrix multiplication, it depends only on the number of operations a given word of data is involved in, so it can be applied to many other computations that have iteration spaces with uneven dimensions.
We believe this will yield new or tighter parallel communication bounds in several cases.

\begin{acks}
This work is supported by the National Science Foundation under Grant No. CCF-1942892 and OAC-2106920.
This project has received funding from the European Research Council (ERC) under the European Union's Horizon 2020 research and innovation program (Grant agreement No. 810367).
\end{acks}

\balance
\bibliographystyle{ACM-Reference-Format}
\bibliography{refs}

\end{document}